\documentclass[12pt]{amsart}
\usepackage{amssymb,amscd}
\usepackage{verbatim}

\usepackage{amsmath,amssymb,graphicx,mathrsfs}   
\usepackage{enumerate}
\usepackage[colorlinks=true,allcolors = blue]{hyperref} 

\usepackage{tikz}
\usetikzlibrary{matrix}

\usepackage[all]{xy}

\let\frak\mathfrak

\def\>{\relax\ifmmode\mskip.666667\thinmuskip\relax\else\kern.111111em\fi}
\def\<{\relax\ifmmode\mskip-.333333\thinmuskip\relax\else\kern-.0555556em\fi}
\def\vsk#1>{\vskip#1\baselineskip}
\def\vv#1>{\vadjust{\vsk#1>}\ignorespaces}
\def\vvn#1>{\vadjust{\nobreak\vsk#1>\nobreak}\ignorespaces}

  \let\ssize\scriptstyle
\let\sssize\scriptscriptstyle

\let\Medskip\medskip
\def\medskip{\par\Medskip}
\let\Bigskip\bigskip
\def\bigskip{\par\Bigskip}

\let\Maketitle\maketitle
\def\maketitle{\Maketitle\thispagestyle{empty}\let\maketitle\empty}

\newtheorem{thm}{Theorem}[section]
\newtheorem{cor}[thm]{Corollary}
\newtheorem{lem}[thm]{Lemma}
\newtheorem{prop}[thm]{Proposition}

\theoremstyle{definition}                                  
\numberwithin{equation}{section}

\theoremstyle{definition}
\newtheorem*{rem}{Remark}

\let\mc\mathcal
\let\nc\newcommand

\let\al\alpha
\let\bt\beta

\let\ka\kappa
\let\la\lambda

\let\phi\varphi
\let\si\sigma

\let\der\partial

\let\ox\otimes

\let\geq\geqslant

\let\leq\leqslant

\let\on\operatorname
\let\bi\bibitem
\let\bs\boldsymbol

\def\C{{\mathbb C}}
\def\Z{{\mathbb Z}}
\def\R{{\mathbb R}}

\def\F{{\mathbb F}}   

\def\+#1{^{\{#1\}}}

\def\Gr{\on{Gr}}

\def\gl{\mathfrak{gl}}

\def\beq{\begin{equation}}
\def\eeq{\end{equation}}
\def\be{\begin{equation*}}
\def\ee{\end{equation*}}

\nc{\bea}{\begin{eqnarray*}}
\nc{\eea}{\end{eqnarray*}}
\nc{\bean}{\begin{eqnarray}}
\nc{\eean}{\end{eqnarray}}

\let\ga\gamma

\nc{\Il}{{\mc I_{\bs\la}}}
\nc{\bla}{{\bs\la}}
\nc{\Fla}{\F_\bla}
\nc{\tfl}{{T^*\Fla}}
\nc{\GL}{{GL_n(\C)}}
\nc{\GLC}{{GL_n(\C)\times\C^*}}

\let\sd s 

\def\ddk_#1{\kk_{#1}\<\>\frac\der{\der\<\>\kk_{#1}}}

\def\bul{\mathbin{\raise.2ex\hbox{$\sssize\bullet$}}}
\def\intt{\mathchoice
{\mathop{\raise.2ex\rlap{$\,\,\ssize\backslash$}{\intop}}\nolimits}
{\mathop{\raise.3ex\rlap{$\,\sssize\backslash$}{\intop}}\nolimits}
{\mathop{\raise.1ex\rlap{$\sssize\>\backslash$}{\intop}}\nolimits}
{\mathop{\rlap{$\sssize\<\>\backslash$}{\intop}}\nolimits}}

\let\kk q 
\let\cc c

\let\Ko K

\def\GZ/{Gelfand-Zetlin}
\def\KZ/{{\slshape KZ\/}}
\def\qKZ/{{\slshape qKZ\/}}
\def\XXX/{{\slshape XXX\/}}

\nc{\A}{{\mc A}}

\nc{\hsl}{\widehat{{\frak{sl}_2}}}

\nc{\BC}{{ \mathbb C}}
\nc{\lra}{\longrightarrow}
\nc{\CO}{{\mathcal{O}}}
\nc{\BZ}{{ \mathbb Z}}
\nc{\hfn}{\hat{\frak{n}}}
\nc\Zs{{\Z/p^s\Z}}
\nc\Zo{{\Zs[z]^0}}
\nc\gr{{\on{gr}}}

\nc\fD{{\frak D}}

\begin{document}

\hrule width0pt
\vsk->

\title{Hypergeometric integrals, hook formulas and Whittaker vectors}

\author[G.\,Felder, A.\,Smirnov, V.\,Tarasov, A.\,Varchenko]
{G.\,Felder$^{\diamond}$, A.\,Smirnov$^\dagger$, V.\,Tarasov$^{\circ}$, A.\,Varchenko$^{\star}$}

\maketitle

\begin{center}
{\it $^{\diamond}$
Department of Mathematics, ETH Zurich, 8092 Zurich, Switzerland}

\vsk.5>
{\it $^\circ\<$Department of Mathematical Sciences,
Indiana University\,--\>Purdue University Indianapolis\kern-.4em\\
402 North Blackford St, Indianapolis, IN 46202-3216, USA\/}

\vsk.5>
{\it $^{\dagger, \star}$ Department of Mathematics, University
of North Carolina at Chapel Hill\\ Chapel Hill, NC 27599-3250, USA\/}

\end{center}

\vsk>
{\leftskip3pc \rightskip\leftskip \parindent0pt \Small
{\it Key words\/}:  Singular vectors, Young diagrams, excited diagrams, hooks, 
 master function, hypergeometric integrals, Whittaker vectors

\vsk.6>
{\it 2020 Mathematics Subject Classification\/}: 
\par}

{\let\thefootnote\relax
\footnotetext{\vsk-.8>\noindent
$^\diamond\<${\sl E\>-mail}:\enspace  giovanni.felder@math.ethz.ch\>,
supported in part by the SNSF under grants 196892, 205607
\\
$^\dagger\<${\sl E\>-mail}:\enspace asmirnov@email.unc.edu\>,
supported in part by the NSF under grant DMS - 2054527 and by
the RSF under grant 19-11-00062
\\
$^\circ\<${\sl E\>-mail}:\enspace vtarasov@iupui.edu\>, 
supported in part by the Simons Foundation under grants 430235, 852996
\\
$^\star\<${\sl E\>-mail}:\enspace anv@email.unc.edu,
supported in part by the NSF under grant DMS-1954266
}}

\begin{abstract}

We determine the coefficient of proportionality between two multidimensional hypergeometric integrals.
One of them is a solution of the dynamical difference equations 
associated with a Young diagram and the other
is the vertex  integral associated with the Young diagram. The coefficient of proportionality
is the inverse of the product of weighted hooks of the Young diagram.
It turns out that this problem is closely related to the question of describing
the action of the center of the universal enveloping algebra of  $\mathfrak{gl}_n$
on the space of Whittaker vectors
in the tensor product of dual Verma modules with fundamental modules, for which we give
an explicit basis of simultaneous eigenvectors.

\end{abstract}

\vsk1.2>
\rightline{\it In memory of Igor Krichever (1950\:--2022)}
\vsk-.9>
\strut

{\small\tableofcontents\par}

\setcounter{footnote}{0}
\renewcommand{\thefootnote}{\arabic{footnote}}

\section{Introduction}

We determine the coefficient of proportionality between two multidimensional hypergeometric integrals.
One of them is a solution of the dynamical difference equations 
associated with a Young diagram and the other
is the vertex integral associated with the Young diagram. The coefficient of proportionality
is the inverse of the product of weighted hooks of the Young diagram. The same coefficient
appears in the problem of diagonalizing the action of the center of the universal enveloping
algebra of $\mathfrak{gl}_n$ on the space of Whittaker vectors in the tensor product of
dual Verma modules with fundamental modules..

\vsk.2>

The standard basis $(u_\la)$ of a fundamental $\gl_n$-module $U_r$ is labeled by
the Young diagrams $\la$ inscribed in an $(n-r)\times r$-rectangle. One assigns to
a basis vector $u_\la$ a system of dynamical difference equations
\bean
\label{dyn}
&&
\\
\notag
&&
 I(z_1, \dots, z_i+\ka, \dots,z_{n-1},\ka) = a_i(z_1,\dots,z_{n-1},\ka)
  I(z_1,\dots,z_{n-1},\ka), \quad i=1,\dots,n-1,
 \eean
where $I(z_1,\dots,z_{n-1},\ka)$ is an unknown scalar function and 
 $a_i$ are suitable coefficients defined in terms of the $\gl_n$-action on $U_r$.
The equations were introduced in \cite{TV}, and solutions were constructed in \cite{MV}.
A solution $I_\la(z,\ka)$ to \eqref{dyn} is  given by a hypergeometric integral of dimension
equal to the number of boxes in $\la$. 

\vsk.2>

We also introduce another hypergeometric integral
$V_\la(z,\ka)$ of the same dimension  associated with $\la$.
We show that it is proportional to $I_{\lambda}(z,\kappa)$ and determine the coefficient of
proportionality between the two integrals in Theorem \ref{thm 4}
\bean
\label{mf}
V_\la(z,\ka) \,=\,\frac 1{\prod_{\square\in\la} h(\square)(z)}\,I_\la(z,\ka)\,,
\eean
where $h(\square)(z)$ is  the hook-weight of a box $\square$ of $\la$, see the definition in
\eqref{hw}.


Our motivation for considering (\ref{mf}) is the following. On the one hand, the enumerative geometry of quiver varieties is controlled by two important objects: the vertex function and the capping operator of a quiver variety \cite{O}. These are the generating functions counting quasimaps to the quiver variety with nonsingular or relative boundary conditions.

On the other hand, with any quiver variety $X$ one can associate a hypergeometric integral. The 3D-mirror symmetry predicts that this integral computes the vertex function of the mirror variety $X^{!}$. The capping operator of $X^{!}$ is obtained in the same way  with additional insertion of the stable envelope functions  to the integral \cite{AO} (also known as weight functions in the theory of qKZ equations). For example, see \cite{SmV1,SmV2} where the integral formulas of this type are discussed in the case of cotangent bundle over Grassmannian $X^{!}=T^{*}Gr(k,n)$.

Let $X$ be the zero-dimensional Nakajima quiver variety of type $A$ associated to a Young diagram  $\la$ \cite{DS1,DS2}. The hypergeometric integral assigned to $X$ is given by the function $V_{\la}(z,\ka)$ and the integral for the capping operator is given by $I_\la(z,\ka)$. Since the cohomology of $X_{\lambda}$ are one-dimensional, it is expected from \cite{O} that both integrals are proportional. Our formula (\ref{mf}) establishes  the coefficient of proportionality explicitly. 

Note, however, that in this case, the mirror $X^{!}_\lambda$ is a 2$|\lambda|$-dimensional variety which can not be realized as a quiver variety. Thus, the vertex function for $X^{!}_\lambda$ is not defined by the methods of \cite{O}. 
To clarify this point,  we refer to the function $V_{\lambda}(z,\kappa)$ as  the vertex integral, instead of the ``vertex function of $X^{!}_{\lambda}$''.

We note also that the coefficient in (\ref{mf}) has a geometric meaning: the mirror variety $X^{!}_\lambda$ is equipped with a torus action with  unique fixed point. The denominator of the coefficient in (\ref{mf}) is the product over the half of the tangent weights at this point with parameters $z_1,\dots,z_n$ understood as the equivariant parameters of the torus \cite{DS2}.


\vsk.2>

To determine the coefficient of proportionality we consider the tensor product $M\ox U_r$\,,\
where $M$ is a Verma module and
analyze singular weight vectors in $M\ox U_r$ of the form
\bea
v(\la) = \sum_{\mu\leq \la} v_\mu\ox u_\mu\,, \qquad \on{with}  \ v_\mu\in M.
\eea
The collection of vectors $(v_\mu)$ is quite a nontrivial object. We simplify it by choosing
a suitable linear function $\psi:M\to\C$ and considering instead
the collection of numbers
$(\psi(v_\mu))$. We develop simple recurrence relations and formulas for these numbers.
We also show that
\bea
V_\la(z,\ka)\big/I_\la(z,\ka) = \psi(v_\emptyset)/ \psi(v_\la).
\eea
Together with formulas for $\psi(v_\mu)$ this equation proves formula \eqref{mf}.

The numbers $\psi(v_\mu)$ are functions of the highest
weight of $M$ associated to the skew Young diagram $\lambda/\mu$.
We show that these numbers arise in
the problem of diagonalizing the action of the center $Z$ of the
universal enveloping algebra on the space of Whittaker vectors of
$M'\otimes U_r$ where $M'$ is the dual of the Verma module $M$.

A Whittaker vector in a $\mathfrak gl_n$-module is a vector on which
the nilpotent subalgebra of lower triangular matrices acts via a fixed
regular character, see Section \ref{sec 3.4}. The center $Z$ acts on
the space of Whittaker vectors $\operatorname{Wh}(M'\otimes U_r)$. A
Whittaker vector $\beta\in M'\otimes U_r$ is uniquely determined by
its contraction $\beta(v)\in U_r$ with the highest weight vector $v$ of $M$.
For generic highest weight of $M$, we show that a basis of
eigenvectors is given by
\[
  \beta_{\lambda}(v)
=
\sum_{\mu\leq \lambda}\sum_{\nu\in E(\lambda/\mu)}
\frac{1}{\prod_{\square\in\lambda\smallsetminus\nu}h(\square)(z)}\,u_\mu,
\]
where $\lambda$ runs over the set of Young diagrams fitting in an $(n-r)\times r$
rectangle and $z$ is an affine function of the highest weight of $M$.
The set $E(\lambda/\mu)$ is the set of Ikeda--Naruse excited diagrams,
which are subsets of $\lambda$ obtained from $\mu$ by moving boxes
according to certain rules. The coefficient of $u_\mu$ is $\psi(v_\mu)$. In particular
the coefficient of $u_\emptyset$ is the coefficient
of proportionality in \eqref{mf} (the coefficient of $u_\lambda$ is normalized to be 1).

\subsection*{Aknowledgements}

The fourth author thanks FIM at ETH Zurich and IHES in Bures-sur-Yvette for hospitality
in June-July 2023. The fourth author also thanks E.\,Mukhin for useful discussions.

\section{Singular vectors}

\subsection{Linear function $\psi$} \label{sec 2.1}

Consider the complex Lie algebra $\gl_n$ with standard generators
$e_{ij}$, $i,j=1,\dots,n$,  simple roots $\al_i$, $i=1,\dots,n-1$, half-sum 
$\rho$ of positive roots. Denote $e_i=e_{i,i+1}$, $f_i=e_{i+1,i}$,  
$h_i=e_{i,i}-e_{i+1,i+1}$ for $i=1,\dots,n-1$.

\vsk.2>

Let $M$ be a $\gl_n$ Verma module with highest weight vector $v$. 
Define a linear function $\psi : M\to \C$ as follows. 
Any vector $v'\in M$ can be written (in general non-uniquely)
 as a finite linear combination of the products of elements $f_1,\dots, f_{n-1}$
 applied to $v$, 
\bea
v' = \sum
c_{i_m, i_{m-1},\dots, i_1} f_{i_m}f_{i_{m-1}}\dots f_{i_1}v\,,
\eea
where $1\leq i_j\leq n-1$ and $c_{i_m, i_{m-1},\dots, i_1} \in\C$.
Set
\bean
\label{psi}
\psi(v') = \sum c_{i_m, i_{m-1},\dots, i_1}\,.
\eean
The function $\psi$ is well-defined since it is zero on Serre's relations
$f_i^2f_{i\pm 1} - 2 f_if_{i\pm 1}f_i + f_{i\pm 1}f_i^2=0$. It is in fact a Whittaker
vector in the dual of $M$, as will be discussed in Section \ref{sec 3.4}.


\subsection{Fundamental representations}\label{sec 2.2}

Let $U_r$, $r=1,\dots,n-1$, be the $r$-th fundamental representation of $\gl_n$.
Its highest weight is $(1,\dots, 1, 0,\dots,0)$ with $r$ ones.
\vsk.2>

The $U_1$ is the vector representation $\C^n$ with standard basis $u_i$, $i=1,\dots,n$.
The $U_r$ is the $r$-th exterior power $\wedge^r \C^n$ of the vector representation
with standard basis 
\bean
\label{wedge w}
u_I := u_{i_1} \wedge u_{i_2} \wedge \dots \wedge u_{i_r} \,,
\eean
where $I=\{i_1< i_2<\dots <i_r\}$ is any  $r$-element subset of $\{1,\dots,n\}$.
Denote by $\mc I_r$ the set of such subsets.

\vsk.2>

The decomposition 
\bea
U_r = \oplus_{I\in\mc I_k}\C u_I\,
\eea
is the weight decomposition. We have $e_{ii}u_I=u_I$ if $i\in I$ and $e_{ii}u_I=0$ otherwise.
Thus, the weight 
$w(u_I)$ 
of $u_I$ is the 
$n$-vector whose $i$-th coordinate is 1 if $i\in I$ and is 0 otherwise.
The vector $u_{I^{min}}$ with $I^{min}=\{1<2 <\dots< r\}$ is a highest weight vector.

\vsk.2>

\subsection{Young diagrams}\label{sec 2.3}

The set  $\mc I_r$ is identified with the set of sequences of nonnegative integers 
\bea
\{0\leq \la_1\leq \dots\leq \la_r\leq n-r\}
\eea
by the formula $\{i_1< i_2<\dots <i_r\}\,\mapsto\, \{i_1-1\leq i_2-2\leq \dots\leq i_r-r\}$.
The set of such sequences is identified with the set of Young diagrams inscribed in the
 $(n-r)\times r$-rectangle $R$.
Thus  the set $\mc I_r$ is identified with the set of Young diagrams inscribed in the rectangle $R$.
 
 For example,
$I^{min}$ corresponds to the empty Young diagram $\emptyset$,
and the rectangle $R$ corresponds to
the subset $\{n-r+1<n-r+2<\dots<n\}$.

\vsk.2>

We conclude that the basis $(u_I)$ of $U_r$ is labeled by the Young diagrams.
 A vector $u_I$ will also be denoted $u_\la$ if $I$ corresponds 
 to a Young diagram $\la$. 
 The weight $w(u_I)$ of $u_I$ will also be denoted by $w(\la)$.

\vsk.2>

The set of Young diagrams is partially ordered with respect to inclusion of the diagrams.
We write  $\mu\leq \la$ if the Young diagram $\la$ contains the Young diagram $\mu$.

\subsection{Singular vectors}\label{sec 2.4}

Let $M$  be the $\gl_n$ Verma module with highest weight $t-\rho$
and highest weight vector $v$\,, where $t=(t_1,\dots,t_n)$.

Fix a Young diagram $\la\in\mc I_r$. Consider the vector subspace $\cap_{i=1}^{n-1}\operatorname{Ker}e_i$ 
of singular vectors in $M\ox U_r$ of weight $w(\la) +t-\rho$. For generic $t$, this space is one-dimensional
with a generator of the form
\bean
\label{main}
v(\la) := v\ox u_\la + \sum_{\mu<\la} v_\mu\ox u_\mu\,
\eean
for suitable vectors $v_\mu\in M$.
Recall the linear function $\psi :M\to\C$. Define the following scalar functions 
$g_{\la/\mu} $ of $t$:
\bean
\label{lm}
g_{\la/\mu} = \psi(v_\mu) \quad \on{for}\quad \mu<\la\quad
\on{and}\quad g_{\la/\la} = 1.
\eean
More precisely, each $g_{\la/\mu}$ is a function of $z=(z_1,\dots,z_{n-1})$
where $z_i=t_{i+1}-t_i$.

\vsk.2>

The main result of this paper is recurrence relations and a formula for functions $g_{\la/\mu}$.

\subsection{Hooks}
\label{sec hooks}

The $(n-r)\times r$-rectangle $R$ lies in the positive quadrant in $\R^2$ and consists of unit
boxes
$\square_{i,j}$, $i=1,\dots, n-r$, $j=1,\dots, r$. The center of a box $\square_{i,j}$ has coordinates
$\big(i- \frac12, j-\frac12\big)$. Every nonempty Young diagram $\la\in\mc I_r$ contains the corner box $\square_{1,1}$.

To every box $\square_{i,j}$ we assign one of $z_1,\dots, z_{n-1}$ by the rule:
\bea
z(\square_{i,j})\, : =\, z_{i-j+r}\,.
\eea
For example, $z(\square_{1,r})=z_1$,
$z(\square_{1,1})=z_r$, $z(\square_{n-r,1})=z_{n-1}$,
$z(\square_{n-r,r})=z_{n-r}$.
We say that $z_{i-j+r}$ is the $z$-label of a box $\square_{i,j}$.

\vsk.2>
Recall that a hook $H_\la(\square_{i,j})$ of a box $\square_{i,j}$ in a Young diagram $\la$
is the set of all boxes $\square_{a,b}$ in $\la$ such that
$a=i$, $b\geq j$,\ or $a\geq i$, $b=j$.
We define the hook-weight of a box $\square$ of $\la$ by the formula
\bean
\label{hw}
h(\square) = 1 + \sum_{\square'\in H_\la(\square)} z(\square')\,.
\eean

\begin{thm}
\label{thm 2.1}

We have
\bean
\label{mu/em}
g_{\la/\emptyset} = \frac 1{\prod_{\square\in\la} h(\square)}\,.
\eean

\end{thm}

For example, if $(r,n)=(2,4)$ and $\la = (2,1)$. Then $\la$ consists of the
three boxes
$\square_{1,1}$, $\square_{1,2}$, $\square_{2,1}$, and
\bea
g_{\la/\emptyset} = \frac 1 {(z_1+z_2+z_3+1)(z_1+1)(z_3+1)}\,.
\eea

\vsk.2>

In Section \ref{sec exc} we give a formula for all coefficients $g_{\la/\mu}$. 
Theorem \ref{thm 2.1} follows from Theorem \ref{thm 3} below.

\subsection{Recurrence relations}\label{sec 2.6}

Let $\la/\mu$ be a skew-diagram. Let $k_i$ be the number of boxes in $\la/\mu$
with $z$-label $z_i$, where $i=1,\dots,n-1$. 
 We put $k_n=0$.
Define the {\it $z$-content} of $\la/\mu$ by the formula
\bea
s_{\la/\mu} = \sum_{i=1}^{n-1} k_i(k_i-k_{i+1} + z_i)\,.
\eea
For example, if $(r,n)=(2,4)$, $\la=(2,2)$, $\mu_1=\emptyset$, $\mu_2=(1)$,
$\mu_3=(1,1)$, $\mu_4=(2)$, $\mu_5=(2,1)$, then
\bea
&&
\phantom{aaa}
s_{\la/\mu_1}=z_1+2z_2+z_3+2,
\qquad
s_{\la/\mu_2}=z_1+z_2+z_3+1,
\\
&&
s_{\la/\mu_3}=z_1+z_2+1,
\qquad
s_{\la/\mu_4}= z_2+z_3+1,
\qquad
s_{\la/\mu_5}=z_2+1.
\eea

\begin{thm}
\label{thm 2}
The following recurrence relations hold:
\bean
\label{rr}
g_{\la/\mu} = \frac 1{s_{\la/\mu}} \sum_{\mu'} g_{\la/\mu'}\,,
\eean
where the sum is over all the Young diagrams $\mu'$ such that
$\mu<\mu'\leq \la$ and the skew-diagram $\mu'/\mu$ consists of one box.

\end{thm}

For example, if $(r,n) = (2,4)$, \, $\la=(2,2)$, and $\mu_i$ are as before, then we have
\bean
\label{ex}
&&
\\
\notag
g_{\la/\emptyset} 
&=&
 \frac 1{z_1+2z_2+z_3+2} \,g_{\la/\mu_2}
= \frac 1{(z_1+2z_2+z_3+2)(z_1+z_2+z_3+1)}\, [ g_{\la/\mu_3}+g_{\la/\mu_4}]
\\
\notag
&&
= \frac 1{(z_1+2z_2+z_3+2)(z_1+z_2+z_3+1)} \Big(\frac1{z_1+z_2+1}
+ \frac1{z_2+z_3+1}\Big) g_{\la/\mu_5}
\\
\notag
&&
= \frac 1{(z_1+z_2+z_3+1)(z_1+z_2+1)(z_2+z_3+1)(z_2+1)}\,
\eean
where in the last step we used $g_{\la/\la}=1$.

\begin{proof}

Denote $\bt =t-\rho - \sum_{i=1}^{n-1} k_i\al_i$ where $k_i$ are some nonnegative integers.
Let $M[\bt]$ be the weight subspace of $M$  of weight $\bt$.

\begin{lem}
\label{lem psi}
For any $v'\in M[\bt]$, we have
\bean
\label{psi v}
\psi ((e_1+ \dots + e_{n-1}) v') \,=\, - \sum_{i=1}^{n-1}k_i(k_i-k_{i+1} + z_i) \,\psi(v').
\eean
\end{lem}

\begin{proof}
The proof is straightforward. 
It is enough to check formula \eqref{psi v} for $v'= f_{m_1}\dots f_{m_k} v$ where $1\leq m_j \leq n-1$
and for any $i$ the sequence $m_1, \dots, m_k$ has exactly $k_i$ elements equal to $i$.

For example, for $\bt =t-\rho - 2\al_1-\al_2$ and $v'= f_1f_1f_2v$ we have
\bea
\psi((e_1+e_2+e_3)v')
&=&
\psi(h_1f_1f_2v + f_1h_1f_2v + f_1f_1h_2v)
\\
&=&
 -\psi((z_1+2) f_1f_2v + z_1f_1f_2v + (z_2+1) f_1f_1v)
= -(2z_1+z_2+3),
\eea
while $\psi(v')=1$.
\end{proof}

To prove the theorem notice that the vector $v(\la)$ is singular and hence
$\psi((e_1+\dots+e_{n-1})v(\la))=0$. By Lemma \ref{lem psi}, we also have
\bean
\label{rel}
\psi((e_1+\dots+e_{n-1})v(\la))= \sum_{\mu<\la}
\Big(-s_{\la/\mu}\,g_{\la/\mu} +\sum_{\mu'}  g_{\la/\mu'}\Big) u_\mu  \,,
\eean
where the second sum is over all the Young diagrams $\mu'$ such that 
$\mu<\mu' \leq \la$ and the skew-diagram $\mu'/\mu$ consists of one box.
Since $(u_\mu)_{\mu\in\mc I_r}$ is a basis of $U_r$, the coefficient of each
$u_\mu$ in \eqref{rel} must be equal to zero. This proves the theorem.
\end{proof}

\begin{cor}
\label{cor 1}
Let $d$ be the number of boxes in $\la/\mu$. Then 
\bean
\label{rr1}
g_{\la/\mu} = 
\sum_{\mu=\mu_1<\mu_2<\dots<\mu_d<\la} \frac1{\prod_{i=1}^d  s_{\la/\mu_i}}\,.
\eean

\end{cor}

See an example in formula \eqref{ex}.

\subsection{Excited diagrams}
\label{sec exc}

Let $\la/\mu$ be a skew-diagram and $D$ a subset of the Young diagram $\la$. 
A box $\square_{i,j}$ of $D$ is called {\it active} if the boxes 
$\square_{i+1,j},\, \square_{i+1,j+1},\,\square_{i,j+1}$ are all in $\la-D$. Let
$b=\square_{i,j}$  be an active box of $D$, define $D_b$
 to be the set obtained by replacing $\square_{i,j}$ in $D$
  by  $\square_{i+1,j+1}$.  We call this replacement an {\it elementary excitation}. An {\it excited diagram} of 
  $\la/\mu$  is a subset of boxes of $\la$ obtained from the Young diagram 
  $\mu$ after a sequence of elementary excitations on active boxes. Let
  $E(\la/\mu)$ be the set of excited diagrams of $\la/\mu$, see this definition in   
 \cite{IN,Na,MPP}.

\begin{thm}
\label{thm 3}
We have
\bean
\label{Nf}
g_{\la/\mu} = \frac 1{\prod_{\square\in\la} h(\square)}\,
\sum_{\nu\in E(\la/\mu)}\prod_{\square \in\nu} h(\square)\,.
\eean

\end{thm}

For example, in the notation of formula \eqref{ex},
the set $E(\la/\mu_2)$ consists of two elements
$\{\square_{1,1}\}$ and $\{\square_{2,2}\}$. Then
\bea
g_{\la/\mu_2} = \frac{z_1+2z_2+x_3+2}
{(z_1+z_2+z_3+1)(z_1+z_2+1)(z_2+z_3+1)(z_2+1)}\,,
\eea
where
\bea
 h(\square_{1,1})+h(\square_{2,2}) = z_1+2z_2+x_3+2.
 \eea

 \begin{rem}
   The equality between \eqref{rr1} and \eqref{Nf} in the case
   $\mu=\emptyset$ is a generalization of the classical hook-length
   formula relating the number of standard Young tableaux of shape
   $\lambda$ to the inverse product of hook-lengths. It converges to
   it in the limit where all $z_i$ are equal and tend to infinity.  In
   the same limit for general $\mu\subset\lambda$ we obtain Naruse's
   generalization for skew diagrams \cite{Na}.
 \end{rem}
 \subsection{Proof of Theorem \ref{thm 3}}
 \label{sec 2.8}

Theorem \ref{thm 3} follows from Corollary \ref{cor 1}  and  Naruse's  formula in \cite{Na}
by a change of parameters $z_1,\dots,z_{n-1}$. 
 
 \vsk.2>
 More precisely, let $\la \in\mc I_r$ be a nonempty Young diagram. We say that a box $\square_{i,j}\in\la$
 is a {\it boundary box} if $\square_{i+1,j+1}\notin\la$.
 Let $\square_{i,j}\in\la$  be a  boundary box. If 
 $\square_{i,j+1}\notin\la$ and $\square_{i+1,j}\notin\la$, then
 $\square_{i,j}$ is an active boundary box according to the definition in
 Section \ref{sec exc} (with $D=\la$).
 If $\square_{i,j+1}\in\la$ and $\square_{i+1,j}\in\la$, 
then we say that $\square_{i,j}$ is a {\it corner boundary box}.
 If $\square_{i,j+1}\in\la$ and $\square_{i+1,j}\notin\la$,
or  if $\square_{i,j+1}\notin\la$ and $\square_{i+1,j}\in\la$,
then we say that  $\square_{i,j}$  is a {\it flat boundary box}.
 
 \vsk.2>
 If $\la$ has a box with $z$-label $z_i$, then $\la$ has an exactly one boundary box
 with $z$-label $z_i$.  We define new parameters $y_i$ by the following formulas.
 We define
 \bea
 z_i &=& y_i -1, 
 \quad
 \on{if\, the \,boundary\, box \,with\, label}\,z_i\,\on{is\, an \,active\, boundary\, box}\,,
 \\
 z_i &=& y_i +1, 
 \quad
  \on{if\, the \,boundary\, box \,with\, label}\,z_i\,\on{is\, a \,corner\, boundary\, box}\,,
 \\
 z_i &=& y_i\,,
 \quad\quad\ \,
  \on{if\, the \,boundary\, box \,with\, label}\,z_i\,\on{is\, a \,flat\, boundary\, box}\,.
\eea
 
 To every box $\square_{i,j}\in\la $ we assign one of $y_1,\dots, y_{n-1}$ by the rule:
\bea
y(\square_{i,j})\, : =\, y_{i-j+r}\,.
\eea
We say that $y_{i-j+r}$ is the $y$-label of a box $\square_{i,j}$.

 \begin{lem}
 \label{lem nc}
 ${}$
 \begin{itemize}
 \item[(i)] Let $\square$ be a box in $\la$ with hook-weight
 $h(\square)(z) = 1 + z_a+z_{a+1}+\dots+z_{b}$ for some $a, b$. Then
 \bean
 \label{hzy}
 h(\square)(z(y)) = y_a+y_{a+1}+\dots+y_{b}\,.
 \eean
 
 \item[(ii)]
 Let $\mu<\la$ and let
 \bea
 s_{\la/\mu}(z)=\sum_{i=1}^{n-1}k_i(k_i-k_{i+1} + z_i)
 \eea
be  the $z$-content of the skew-diagram $\la/\mu$. Then
 \bean
 \label{ciny}
 s_{\la/\mu}(z(y)) =\sum_{i=1}^{n-1}k_iy_i\,.
 \eean 
 \end{itemize}
 \end{lem}

 For example,  in the notation of formula \eqref{ex},
the change of variables for $\la$ is
\bea
z_1=y_1,\qquad 
z_2=y_2 -1,\qquad 
 z_3=y_3\,.
 \eea
 Then $s_{\la/\emptyset}(z) = z_1+2z_2+z_3+2$ and
 $s_{\la/\emptyset}(z(y)) = y_1+2y_2+y_3$\,. \ Similarly,
 $s_{\la/\mu_2}(z) = z_1+z_2+z_3+1$ and
 $s_{\la/\mu_2}(z(y)) = y_1+y_2+y_3$\,.

 \begin{proof}
The proof of the lemma is straightforward.  For example, we prove part (i).
Let $\square_{i,j}$ be a box in $\la$ with hook-weight
$h(\square_{i,j})(z) = 1 + z_a+z_{a+1}+\dots+z_{b}$ for some $a, b$. 
The boxes 
$\square_{i,a}$ and $\square_{b,j}$ are boundary boxes of $\la$. Let us walk from
the box $\square_{i,a}$ to the box $\square_{b,j}$ through the boundary boxes of $\la$.
This walk consists of $b-a+1$ boundary boxes with  $z$-labels  $z_a, z_{a+1}, \dots, z_b$.
Let $\ell$ be  the number of active boundary boxes in this walk. Then the walk  has
 exactly 
 $\ell-1$ corner boundary boxes.
Hence our change of variables transforms 
$h(\square)(z) = 1 + z_a+z_{a+1}+\dots+z_{b}$ to 
$1 + y_a+y_{a+1}+\dots+y_{b} -\ell +(\ell-1) = y_a+y_{a+1}+\dots+y_{b}$. Part (i) is proved.
 \end{proof}
 
 Having  Lemma \ref{lem nc} we rewrite Corollary \ref{cor 1} in terms of the variables $y_i$. 
 Namely, define the $y$-hook-weight of a box $\square\in\la$ by the formula
\bea
\tilde h(\square) = \sum_{\square'\in H_\la(\square)} y(\square')\,,
\eea
and the $y$-content of a skew-diagram $\la/\mu$ by the formula
\bea
\tilde s_{\la/\mu} = \sum_{i=1}^{n-1} k_iy_i\,,
\eea
if $k_i$ is the number of boxes in $\la/\mu$
with $y$-label $y_i$.
 Then
 \bea
 \tilde h(\square)(y)= h(\square)(z(y)), \qquad 
 \tilde s_{\la/\mu}(y) = s_{\la/\mu}(z(y))
 \eea
 by Lemma   \ref{lem nc}. Formula \eqref{rr1} takes the form:
 \bean
\label{rr2}
g_{\la/\mu}(z(y)) 
&=& 
\sum_{\mu=\mu_1<\mu_2<\dots<\mu_d<\la} \frac1{\prod_{i=1}^d  \tilde s_{\la/\mu_i}(y)}\,.
\eean
On the other hand, H.\,Naruse's formula \cite[page 13]{Na} states that
\bean
\sum_{\mu=\mu_1<\mu_2<\dots<\mu_d<\la} \frac1{\prod_{i=1}^d  \tilde s_{\la/\mu_i}(y)}
&=&
\frac 1{\prod_{\square\in\la} \tilde h(\square)(y)}\,
\sum_{\nu\in E(\la/\mu)}\prod_{\square \in\nu} \tilde h(\square)(y)\,,
\eean
see also \cite{IN, MPP}.  Hence,
\bean
g_{\la/\mu}(z(y)) 
&=& 
\frac 1{\prod_{\square\in\la} h(\square)(z(y))}\,
\sum_{\nu\in E(\la/\mu)}\prod_{\square \in\nu} h(\square)(z(y))\,,
\eean
and Theorem \ref{thm 2} is proved.

\subsection{Change of variable and weight shift}\label{sec 2.9}
The change of variables $y\mapsto z=z(y)$ can be understood in terms of weights
  as follows:
\begin{lem}\label{l-001} Let $\zeta\colon \mathbb C^n\to \mathbb C^{n-1}$
  be the linear map $t\mapsto (t_2-t_1,\dots,t_{n}-t_{n-1})$. If $z=\zeta(t)$  then $y=\zeta(t-w(\lambda))$.
\end{lem}
\begin{proof} The weight corresponding to the Young diagram $\lambda$
  is $w(\lambda)=(\epsilon_1,\dots,\epsilon_n)$ where
  $\epsilon_i\in\{0,1\}$ with $\epsilon_i=1$ iff
  $i\in\{i_1<\dots< i_r\}$ where $i_k=\lambda_k+k$ ($k=1,\dots,r$), see
  Section \ref{sec 2.3}. Then for each $i=1,\dots,n-1$, we have
  \begin{itemize}
  \item $\epsilon_{i+1}-\epsilon_{i}=-1$ if $i=i_k$ for some
    $k\in\{1,\dots,r\}$ and $i_{k}<i_{k+1}-1$, where we set
    $i_{r+1}=n+1$,
  \item $\epsilon_{i+1}-\epsilon_{i}=1$ if $i+1=i_k$ for some
    $k\in\{1,\dots,r\}$ and $i_{k-1}<i_k-1$, where we set $i_0=0$,
  \item $\epsilon_{i+1}-\epsilon_{i}=0$, otherwise.
  \end{itemize}
  The first alternative occurs iff $i=i_k$ and
  $\lambda_{k}<\lambda_{k+1}$ where we set $\lambda_{r+1}=n-r$. This
  is exactly the condition for the box with coordinates
  $(\lambda_k,r-k)$, which has $z$-label $z_i$, to be an active
  boundary box.
  
  The second alternative occurs iff $i+1=i_k$ and $\lambda_{k-1}<\lambda_k$
  where we set $\lambda_0=0$. This is exactly the condition for the
  boundary box with coordinates $(\lambda_k,r-k+1)$, which has $z$-label
  $z_i$, to be a corner boundary box.
\end{proof}

 \section{Applications}
 
 \subsection{Master function}
 
 Let $\la$ be a Young diagram inscribed in the $(n-r)\times r$ rectangle $R$.
  Let $\la$  have $k_i$ boxes with $z$-label $z_i$ for $i=1,\dots,n-1$. Denote
 $k= k_1+\dots+k_{n-1}$.
 
 \vsk.2>

 Consider $\C^k$ with coordinates 
  $x=(x_{i,j})$, $i=1,\dots,n-1$, $j=1,\dots,k_i$. Define the {\it master function}
  \footnote{The {\it superpotential} in the terminology of enumerative geometry.}
 \bean
 \label{master}
 &&
 \\
 \notag
 &&
\Phi_\la(x,z)= \prod_{i=1}^{n-1}\prod_{j=1}^{k_i} x_{i,j}^{z_i+1} 
 \prod_{j=1}^{k_r} (x_{k,j}-1)^{-1}
 \prod_{i=1}^{n-1} \prod_{j<j'} (x_{i,j}-x_{i,j'})^2
 \prod_{i=1}^{n-2} \prod_{j=1}^{k_i} \prod_{j'=1}^{k_{i+1}}
 (x_{i,j}-x_{i+1,j'})^{-1}\,.
 \eean
 The linear functions $x_{i,j}$, $x_{k,j}-1$, $x_{i,j}-x_{i,j'}$,
 $x_{i,j}-x_{i+1,j'}$ appearing in the master function define
 an arrangement $\mc C$ of hyperplanes in $\C^k$.
 
 \vsk.2>
 
 The group $G=S_{k_1}\times \dots \times S_{k_{n-1}}$ acts on
 $\C^k$ by permuting the coordinates $(x_{i,j})$ with the same first index $i$.
 The arrangement $\mc C$ and master function $\Phi_\la(x,z)$ are $G$-invariant.
 
 \vsk.2>
 
 For $\ka\in\C^\times$, the multivalued function $\Phi_\la(x,z)^{1/\ka}$ defines a rank one
 local system $\mc L_\ka$  on the complement $X=\C^k\setminus \mc C$ to the arrangement. 
 The group $G$ acts on the homology $H_*(X;\mc L_\ka)$ and 
 cohomology $H^*(X;\mc L_\ka)$.
 Let $H_k(X;\mc L_\ka)^-\subset H_k(X;\mc L_\ka)$ 
  and $H^k(X;\mc L_\ka)^-\subset H^k(X;\mc L_\ka)$
  be the isotypical components corresponding to 
 the sign representation. It is known that for generic $\ka$, we have $\dim H^k(X;\mc L_\ka)^- 
 =\dim H_k(X;\mc L_\ka)^-
  =1$ 
 since the space $H^k(X;\mc L_\ka)^-$  can be identified with the space of singular vectors in
 $M\times U_r$ of weight $w(\la) +t-\rho$, which is of dimension 1, see \cite{SV}.
 
 \subsection{Weight function of $u_\la$}

Let $u_\la$ be the basis vector of $U_r$ corresponding to the diagram $\la$. The 
 vector $u_\la$ is related to the highest weight vector $u_\emptyset$
  by the formula
\bean
\label{f-pres}
u_\la = f_{\ell_{k}}\dots f_{\ell_2}f_{\ell_1} u_\emptyset\,,
\eean
where $ f_{\ell_{k}},\dots,f_{\ell_2}, f_{\ell_1}$ is a certain ({\it admissible}) sequence of Cartan generators $f_1,\dots,f_{n-1}$
in which there are exactly $k_i$ elements $f_i$ for every $i=1,\dots,n-1$. Let $\mc F_\la$ 
be the set of all such admissible sequences.

\vsk.2>

Let $f=\{f_{\ell_k},\dots, f_{\ell_{1}}\}$ be an admissible sequence. Define the 
function $W^\circ_{f}(x)$,
\bean
\label{Wo}
W^\circ_{f}(x) = \frac 1{ 
(x_{a_{k},b_{k}}-x_{a_{k-1},b_{k-1}})  \dots  (x_{a_3,b_3}-x_{a_2,b_2})
(x_{a_2,b_2}-x_{a_1,b_1})(x_{a_1,b_1}-1)}
\eean
such that 
\begin{enumerate}
\item[(i)]
each variable $x_{i,j}$ is present in \eqref{Wo},
\item[(ii)]
if  $(x_{a_{c},b_{c}}-x_{a_{c-1},b_{c-1}})$ is any of the factors, then  $a_c =\ell_c$,
\item[(iii)]
 for any $i$ and $1\leq j<j'\leq k_i$, the variable $x_{i,j}$ appears in \eqref{Wo} on the right from the variable
$x_{i,j'}$\,.
\end{enumerate}
These properties determine the function $W^\circ_{f}(x)$ uniquely.

Define 
\bea
W_\la(x) 
=\on{Sym}_{x_{1,1},\dots,x_{1,k_1}}
\dots \on{Sym}_{x_{n-1,1},\dots,x_{n-1,k_{n-1}}}\Big[ \sum_{f\in\mc F_\la}
W^\circ_{f}(x)\Big] \,,
\eea
where we use the notation
$\on{Sym}_{t_1,\dots,t_j}P({t_1,\dots,t_j}) := \sum_{\si\in S_j} P(t_{\si(1)},\dots,t_{\si(j)})$.
The function $W_\la(x)$ is called the {\it weight function} of the vector
$v\ox u_\la$ in $M\ox U_r$.

\vsk.2>

For example, if $(r,n) = (2,4)$, $\la=(2,2)$, then 
$x=(x_{1,1},x_{2,1},x_{2,2},x_{3,1})$. There are two admissible sequences
\bea
u_\la =f_2f_3f_1f_2u_\emptyset = f_2f_1f_3f_2u_\emptyset \,,
 \eea
 and
 \bea
 W_\la(x) &=&
 \frac 1 {(x_{2,2}-x_{3,1}) (x_{3,1}-x_{1,1})(x_{1,1}-x_{2,1})(x_{2,1}-1)}
\\
&+&
 \frac 1 {(x_{2,1}-x_{3,1}) (x_{3,1}-x_{1,1})(x_{1,1}-x_{2,2})(x_{2,2}-1)}
\\
&+&
 \frac 1 {(x_{2,2}-x_{1,1}) (x_{1,1}-x_{3,1})(x_{3,1}-x_{2,1})(x_{2,1}-1)}
\\
&+&
 \frac 1 {(x_{2,1}-x_{1,1}) (x_{1,1}-x_{3,1})(x_{3,1}-x_{2,2})(x_{2,2}-1)}\,.
 \eea
 
 \subsection{Two integrals}

 Let $\ga \in H_k(X;\mc L_\ka)^-$ be a generator. 
Let
\bea
 \wedge_{i,j} dx_{i,j}\,
\eea
denote the wedge product in the lexicographic order
of all the differentials $dx_{i,j}$\,.\ Define two functions 
 \bea
 I_\la(z,\ka) = \int_\ga \Phi_\la(x,z)^{1/\ka} W_\la(x) \big(\wedge_{i,j} dx_{i,j}\big)\,,
 \qquad
 V_\la(z,\ka) = \int_\ga \Phi_\la(x,z)^{1/\ka} \frac 1{\prod_{i,j} x_{i,j}} \big(\wedge_{i,j} dx_{i,j}\big)\,.
 \eea
 Both function are multiplied by the same nonzero 
 constant if we choose a different generator.
 
 \vsk.2>
As shown in \cite{MV},  the first function is a {\it hypergeometric solution} of the dynamical difference
equations associated with  the weight subspace $U_r[w(\la)]$ of the
$\gl_n$-module $U_r$.
The dynamical equations were introduced in \cite{TV}. The
(hypergeometric) solutions of the dynamical equations
were constructed in \cite{MV}. The dynamical equations is a system
 of difference equations of the form
\bea
 I(z_1, \dots, z_i+\ka, \dots,z_{n-1},\ka) = a_i(z_1,\dots,z_{n-1},\ka)
  I(z_1,\dots,z_{n-1},\ka), \qquad i=1,\dots,n-1,
 \eea
 for suitable coefficients $a_i$ defined in terms of the $\gl_n$-action on $U_r$.

 \vsk.2>
 
 We call the second function $V_\la(z,\ka)$\,-- the  {\it vertex  integral} associated with 
 the weight subspace $U_r[w(\la)]$ of the $\gl_n$-module $U_r$.

 \begin{thm}
 \label{thm 4}
 
 We have
 \bean
 \label{IV}
 V_\la(z,\ka)  \,=\,\frac 1{\prod_{\square\in\la} h(\square)(z)}\, I_\la(z,\ka)\,.
 \eean
 
 \end{thm}
 
 The starting goal of this project was to find the coefficient of proportionality between
the vertex integral $ V_\la(z,\ka)$ and the hypergeometric solution
$ I_\la(z,\ka)$ which turned out to be the inverse of the
product of the hook-weights of the boxes of the Young diagram $\la$.
 
 \begin{proof}
 
 In \cite{SV}, given $M\ox U_r$ and $\la\in\mc I_r$, a vector $\bar v(\la)$ is constructed,
 \bea
\bar v(\la) := \bar v_\la \ox u_\la + \sum_{\mu<\la} \bar v_\mu\ox u_\mu\,,
\qquad
\bar v_\la, \bar v_\mu \in H^k(X;\mc L)^-\ox M.
\eea
Thus $\bar v(\la) \in H^k(X;\mc L_\ka)^-\ox M\ox U_r$. 
The vector $\bar v(\la)$ 
 has $\gl_n$-weight $w(\la)+t-\rho$ 
and is singular with respect to the factors $M\ox U_r$. 
The vector $\bar v(\la)$ is a cohomological version of the vector
$v(\la)$ defined in \eqref{main} and studied in the previous sections.
\vsk.2>

The vector $\bar v_\la$ is represented by the differential form 
\bea
\big(\Phi(x,z)^{1/\ka}W_\la(x) \big(\wedge_{i,j} dx_{i,j}\big)\big) \ox v,
\eea
see \cite{SV}.

The vector $\bar v_\emptyset$ is represented by a differential form constructed
as follows. A sequence 
$ f_{\ell_{k}},\dots,f_{\ell_2}, f_{\ell_1}$ is called  {\it weakly admissible}
if for $i=1,\dots, n-1$, the sequence contains exactly $k_i$ elements $f_i$.
Let $\mc F^\star_\la$ 
be the set of all weakly admissible sequences.

\vsk.2>

For example, if $(r,n)=(2,4)$ and $\la=(2,2)$, then $\mc F^\star_\la$ consists of 12 sequences:
$\{f_2,f_2, f_1, f_3\}$, \,\dots, \,$\{f_3,f_1, f_2, f_2\}$.

\vsk.2>

Let $f=\{f_{\ell_k},\dots, f_{\ell_{1}}\}$ be a weakly admissible sequence.
 Define the function $W^\star_{f}(x)$ by the formula
\bean
\label{Wd}
W^\star_{f}(x) = \frac 1{ 
(x_{a_{k},b_{k}}-x_{a_{k-1},b_{k-1}})  \dots  (x_{a_3,b_3}-x_{a_2,b_2})
(x_{a_2,b_2}-x_{a_1,b_1})x_{a_1,b_1}}
\eean
such that
\begin{enumerate}
\item[(i)]
each variable $x_{i,j}$ is present in \eqref{Wd},
\item[(ii)]
if  $(x_{a_{c},b_{c}}-x_{a_{c-1},b_{c-1}})$ is any of the factors, then  $a_c =\ell_c$,

\item[(ii')]
$(a_1,b_1) = (\ell_1,1)$,
\item[(iii)]
 for any $i$ and $1\leq j<j'\leq k_i$, the variable $x_{i,j}$ appears in \eqref{Wo} on the right from the variable
$x_{i,j;}$.
\end{enumerate}
These properties determine the function $W^\star_{f}(x)$ uniquely. 
\vsk.2>

Notice that the last factors in \eqref{Wo} and \eqref{Wd} are different.

\vsk.2>

 Define the function $W_{f}(x)$ by the formula
\bean
\label{Wf}
W_f(x) 
=\on{Sym}_{x_{1,1},\dots,x_{1,k_1}}
\dots \on{Sym}_{x_{n-1,1},\dots,x_{n-1,k_{n-1}}}\big[ W^\star_{f}(x)\big] \,.
\eean

Then the vector $v_\emptyset$ is represented by the differential form
\bea
\sum_{f=\{f_{\ell_k},\dots, f_{\ell_{1}}\}\in \mc F^\star_\la} 
\big( \Phi_\la(x,z)^{1/\ka}W_f(x)\big(\wedge_{i,j} dx_{i,j}\big)\big) \ox f_{\ell_k}\dots f_{\ell_{1}}v ,
\eea
see \cite{SV}.

Let $\ga\in H_k(X;\mc L_\ka)^-$ be a generator. The integral
of $\bar v(\la)$ over $\ga$ is a scalar multiple the vector $v(\la)$,
\bea
\int_\ga\bar v(\la) = c(z,\ka)\,v(\la) .
\eea
We apply the linear function $\psi:M\to \C$ to both sides of this equation and equate the 
coefficients of $u_\la$ and $u_\emptyset$. Then
\bea
c (z,\ka) &=& \int_\ga
\Phi(x,z)^{1/\ka}W_\la(x) \big(\wedge_{i,j} dx_{i,j}\big),
\\
c(z,\ka) \,g_{\la/\emptyset} (z) &=& \int_\ga
\Phi(x,z)^{1/\ka}  
\Big(\sum_{f\in \mc F^\star_\la} W_f(x)\Big)
\big(\wedge_{i,j} dx_{i,j}\big)\,.
\eea
Using the formula
\bea
\sum_{\si\in S_k}\frac 1{(s_{\si(k)}- s_{\si(k-1)}) (s_{\si(k-1)}- s_{\si(k-2)})\dots (s_{\si(2)}- s_{\si(1)})s_{\si(a)}}=
\frac 1{\prod_{j=1}^k s_j}
\eea
and the definition of $W_f(x)$ we conclude that
\bea
\sum_{f\in \mc F^\star_\la} W_f(x) = \frac 1{\prod_{i,j}x_{i,j}}\,.
\eea
Hence
\bea 
c (z,\ka) = I_\la(z,\ka) ,
\qquad
c(z,\ka) \,g_{\la/\emptyset} (z) = V_\la(z,\ka).
\eea
Now formula \eqref{mu/em} implies Theorem \ref{thm 4}.
 \end{proof}
 

\subsection{Whittaker vectors}\label{sec 3.4}
Let $\mathfrak n^{-}$ be the maximal nilpotent subalgebra of $\mathfrak{gl}_n$
of lower triangular matrices. It is generated by $f_1,\dots, f_{n-1}$. Let $\eta\colon
\mathfrak n^-\to\mathbb C$ be the character of the Lie algebra $\mathfrak n^-$
such that  $\eta(f_i)=-1$  for all $i$.

A {\em Whittaker vector} in a $\mathfrak {gl}_n$-module
$V$ is a vector $u\in V$ so that $xu=\eta(x)u$ for all $x\in\mathfrak
n^-$. This notion was introduced and studied by B. Kostant, \cite{Ko}.
The space of Whittaker vectors in $V$ is denoted
$\operatorname{Wh}(V)$. It is a module over the center $Z$ of the
universal enveloping algebra of $\mathfrak{gl}_n$. A Whittaker vector
$u\neq0$ such that $zu=\chi(z)u$ for all $z\in Z$ and some character
$\chi\colon Z\to\mathbb C$ is said to have infinitesimal character $\chi$.

For example let $M'=\mathrm{Hom}_{\mathbb C}(M,\mathbb C)$ be the dual of a
Verma module $M$.  
It is a $\mathfrak {gl}_n$-module
for the action $(x\alpha)(m)=-\alpha(xm)$, $\alpha\in M$, $m\in M$,
$x\in\mathfrak{gl}_n$. Central elements $z\in Z$ act on $M$ as multiples
$\chi_{M'}(z)$ of the identity, for some character $\chi_{M'}\colon Z\to \mathbb C$.
The linear function $\psi\in M'$ of Section \ref{sec 2.1} is defined
by the conditions $f_i\psi=-\psi$ and $\psi(v)=1$ and is in particular a
Whittaker vector. On the other hand, any Whittaker
vector in $M'$ is uniquely determined by its value on $v$ since $v$ generates
$M$ as a module over $U(\mathfrak n^-)$. Thus:

\begin{lem}
  The space of Whittaker vectors
  $\operatorname{Wh}(M')$ is one-dimensional, spanned by the Whittaker
  vector $\psi$ of infinitesimal weight $\chi_{M'}$.
\end{lem}

More generally let us consider the problem of describing  the $Z$-module
of Whittaker
vectors in the $\mathfrak{gl}_n$-module
$M'\otimes U\cong
\operatorname{Hom}_{\mathbb C}(M,U)$ for a Verma module $M$ and
a fundamental module $U$. By definition $\alpha\colon M\to U$ is
a Whittaker vector if and only if
\[
  x\alpha(m)=\alpha(xm)+\eta(x)\alpha(m),\quad \forall m\in M, x\in\mathfrak n^-.
\]
It follows that a Whittaker vector
$\alpha$ is again uniquely determined by its value on the highest weight
vector $v\in M$.

Let $M_{t-\rho}$ denote the Verma module of highest weight $t-\rho\in\mathbb C^n$.
Let $\chi(t)=\chi_{M'_{t-\rho}}$ be the infinitesimal character of its dual.
\begin{prop}\label{p-001}
  Let  $r\in\{1,\dots, n-1\}$ and $t\in\mathbb C^n$ be generic and
  set $z_i=t_{i+1}-t_i$, $(i=1,\dots,n-1)$.
  Then for each Young diagram $\lambda\in \mathcal I_r$ there is a unique
  Whittaker vector
  $\alpha_{\lambda,t}\in \operatorname{Hom}_{\mathbb C}(M_{w(\lambda)+t-\rho},U_r)$
  of infinitesimal character $\chi(t)$ such that
  \[
    \alpha_{\lambda,t}(v)=\sum_{\mu\leq \lambda}g_{\lambda/\mu}(z)u_\mu,
  \]
  where 
\[
  g_{\lambda/\mu}(z)=\frac{1}{\prod_{\square\in\lambda}h(\square)}
  \sum_{\nu\in E(\lambda/\mu)}\prod_{\square\in\nu}h(\square),
\]
and $h(\square)=1+\sum_{\square'\in H_\lambda(\square)}z(\square')$.
\end{prop}

\begin{proof} The morphism of $\mathfrak{gl}_n$-modules
  $M_{w(\lambda)+t-\rho}\to M_{t-\rho}\otimes U_r$ sending the highest
  weight vector to the singular vector
  $v(\lambda)=\sum_{\mu\leq \lambda} v_\mu\otimes u_\mu$, see Section
  \ref{sec 2.4}, induces a morphism
  \[
    M'_{t-\rho}\to \operatorname{Hom}_{\mathbb C}(M_{w(\lambda)+t-\rho}, U_r).
  \]
  The morphism property implies that it sends the Whittaker vector
  $\psi$ of infinitesimal character $\chi(t)$ to a Whittaker vector
  $\alpha$ with the same infinitesimal character. By definition
  $\alpha(v)=\sum_{\mu\leq \lambda} \psi(v_\mu)u_\mu$ and
  $g_{\lambda/\mu}=\psi(v_\mu)$ is given in Theorem \ref{thm 3}.
\end{proof}

We thus obtain an explicit diagonalization of the action of the center $Z$ on the
space of Whittaker vectors in $M'\otimes U$ for a generic Verma module
$M$ and a fundamental module $U$:
\begin{thm}
  Let $t\in\mathbb C^n$, $r\in\{1,\dots, n-1\}$ and $y_i=t_{i+1}-t_i$ $(i=1,\dots,n-1)$.
  Then
  \[
    W=\operatorname{Wh}(\operatorname{Hom}_{\mathbb C}(M_{t-\rho},U_r))
 \]
 has dimension $\operatorname{dim}(U_r)$. For generic $t$, $W$
 decomposes into a direct sum
 $W=\oplus_{\lambda\in\mathcal I_r} W_{\lambda}$ of $Z$-invariant
 one-dimensional subspaces on which $Z$ acts by the character
 $\chi(t-w(\lambda))$.  The subspace $W_\lambda$ is spanned by the
 Whittaker vector $\beta_{\lambda,t}$, such that
 $\beta_{\lambda,t}(v)=\sum_{\mu\leq \lambda}\tilde
 g_{\lambda/\mu}(y)u_\mu$ with
 \[   
  \tilde g_{\lambda/\mu}(y)=\frac{1}{\prod_{\square\in\lambda}\tilde h(\square)}
  \sum_{\nu\in E(\lambda/\mu)}\prod_{\square\in\nu}\tilde h(\square),
\]
and $\tilde h(\square)=\sum_{\square'\in H_\lambda(\square)}y(\square')$.
\end{thm}

\begin{proof}
  By Proposition \ref{p-001} the Whittaker vectors
  $\beta_{\lambda,t}=\alpha_{\lambda,t-w(\lambda)}$ belong to $W$ and
  have infinitesimal character $\chi(t-w(\lambda))$.  Since
  $\beta_{\lambda}(v)=u_\lambda$ plus a linear combination of $u_\mu$
  with $\mu<\lambda$, these vectors are linearly independent. We need
  to show that they span $W$. Let $\beta$ is a Whittaker vector. Since
  the vectors $\beta_\lambda(v)$ form a basis of $U_r$ there exist
  coefficients $c_i\in\mathbb C$ so that
  $\gamma=\beta-\sum_{\lambda\in \mathcal I_r} c_\lambda\beta_\lambda$
  is a Whittaker vector which vanishes on $v$. Since a Whittaker
  vector is uniquely determined by its value on $v$, $\gamma$ must be
  zero.

  Let $t'=t-w(\lambda)$ and $z_i'=t'_{i+1}-t'_i$, $(i=1,\dots,n-1)$.
  We need to compute the coeffcients $g_{\lambda/\mu}(z')$.
  By Lemma \ref{l-001}, $z'=z(y)$ defined in Section \ref{sec 2.8}
  and therefore $g_{\lambda/\mu}(z')=
  \tilde g_{\lambda/\mu}(y)$.
\end{proof}  


\end{document}